\documentclass[rewiew]{elsarticle}

\usepackage{lineno,hyperref,amssymb}

\usepackage{amssymb}
\usepackage{amsmath}
\usepackage{latexsym}
\usepackage{epsfig}
\usepackage{graphicx}
\usepackage{color}
\usepackage{fouridx}

\makeatletter

\renewcommand*{\p@subsection}{}

\renewcommand*{\p@subsubsection}{}
\makeatother

\numberwithin{equation}{section}
\newtheorem{prop}{Proposition}

\newtheorem{proof}{Proof}

\begin{document}

\begin{frontmatter}

\title{Fractional Floquet theory
\\
{\small\bf   [ Chaos, Solitons \& Fractals 168 (2023) 113196 ]}}

\author{Alexander Iomin}
\ead{iomin@physics.technion.ac.il}
\address{Department of Physics, Technion, Haifa, 32000, Israel, \\
Max Planck Institute for the Physics of Complex Systems, 01187 Dresden, Germany}

\begin{abstract}

A fractional generalization of the Floquet theorem is suggested for 
fractional Schr\"odinger equations (FTSE)s with the time-dependent periodic Hamiltonians.
The obtained result,  called the fractional Floquet theorem (fFT), is formulated in the form of 
the Mittag-Leffler function, which is considered as the eigenfunction of the Caputo fractional derivative.
The suggested formula makes it possible to reduce the FTSE to the standard quantum mechanics with 
the time-dependent Hamiltonian, where the standard Floquet theorem is valid.
Two examples related to quantum resonances are considered 
as well to support the obtained result.


\end{abstract}


\end{frontmatter}

\section{Introduction}

In this paper, we consider a Floquet theory for fractional time Schr\"odinger equations.
Separately, these two issues of both the Floquet theorem and fractional quantum mechanics 
are well defined and well reviewed.
The Floquet theory is introduced to treat linear differential equations with time-periodic coefficients
\cite{teschl} and it is widely used in description of driven quantum systems such as 
the interaction of radiation with matter, quantum nonlinear resonances, quantum chaos and so on \cite{GrHa98,santoro,CaMo89,satija}. The Floquet theorem states that if the time-dependent Hamiltonian 
$\hat{H}(t)=\hat{H}(t+T)$ is periodic in time with the period $T$, then the wave function, as the solution to the corresponding Schr\"odinger equation, has the form 
$\psi(t)=e^{-i\epsilon t}u(t)$. Here $u (t+T)=u(t)$ is the periodic eigenfunction of a so called 
Floquet operator $\hat{F}$,
\begin{equation}\label{int1}
\hat{F}u(t)\equiv[\hat{H}(t)-i\hbar\partial_t]u(t)=\epsilon u(t),
\end{equation}
where $\hbar$ is an effective dimensionless  Planck constant and $\epsilon$ is the quasienergy spectrum  
\cite{howland1,howland2} (see also Ref. \cite{santoro}). It is worth  mentioning that in many cases, 
Eq. \eqref{int1} has no analytic solutions, like in quantum chaos \cite{CaMo89}. Moreover, the spectrum can have fractal band structures, such as the Hofstadter butterfly \cite{KeKrSpGe2000,HuWeIoKeFiGe2000}, see also an extended review \cite{satija}. This situation is also reflected in the exploration of the Floquet engineering of quantum materials \cite{Hol2016,BuDAPo2015}.

Fractional time Sch\"odinger equations belong to the field of fractional quantum mechanics.
The latter has been introduced in Refs. \cite{las1,las2} and currently it is a well established field of quantum mechanics \cite{las2017,taras2008,iom2019,taras2021}, which is
also supported by experimental evidences \cite{WuHuChChHs2010,Longhi,BaBeWi2008}.
The fractional time Schr\"odinger equation (FTSE) has been introduced by analogy with a fractional diffusion equation \cite{naber},
where the time derivative is replaced by the fractional time derivative, namely $i\hbar\partial_t \rightarrow (i\hbar)^{\alpha}\partial_t^{\alpha}$ with  $\alpha\in (0,2)$, where
 $\partial_t^{\alpha}$ is the Caputo fractional derivative,
\begin{equation}\label{int2}
\partial_t^{\alpha}\equiv \fourIdx{C}{t_0}{\alpha}{t}{D}f(t)= 
\frac{1}{\Gamma(\alpha)}\int_{t_0}^t\frac{d^nf(\tau)}{d\tau^n}(t-\tau)^{\alpha-n}d\tau,
\end{equation}
with $n-1<\alpha<n$. In the present study, we consider $\alpha<1, ~(n=1)$ and $t_0=0$.
An alternative replacement  of the time derivative, 
$i\hbar\partial_t \rightarrow i(\hbar)^{\alpha}\partial_t^{\alpha}$ has been suggested as well in Refs.
\cite{iom2019,AcYaHa2013} with the main argument that such ``fractionalization''  of the time derivative does not lead to any artificial non-physical effects. It is worth noting that the latter approach to the fractional time derivative in the Schr\"odinger equation is based on the generalized Taylor series that has been proposed in Ref. \cite{TrRiBo1999}, and relates  to a general memory in the system \textit{e.g.}, in the form of a fading memory  
\cite{tarasov2018}.
An exact example of the relation between a fractional diffusion equation and its quantum counterpart  is studied in Ref. \cite{iom2020}.

So far, the extended studies of the FTSE  are restricted  by consideration of  conservative systems, 
where the Hamiltonian is time independent.
The main reason for this restriction is that the Caputo fractional derivative destroys the
periodicity of any periodic functions, see \ref{app-A}. Another important obstacle of the treatment of the FTSE with the time-dependent Hamiltonians is the violation of the general Leibniz rule by fractional derivatives, see \textit{e.g.}, Ref.  \cite{SaKiMa93}.

In this paper, the FTSEs with the time-periodic operators, $\hat{H}(t+T)=\hat{H}(t)$ are considered. We also suggest the extension of the Floquet theorem for the FTSE, 
\begin{equation}\label{int3}
\partial_t^{\alpha}\Psi(t) =\hat{H}(t)\Psi(t) .
\end{equation}
We call this theorem by the ``fractional Floquet theorem'' (fFT). 
To apply this theorem for the various physical systems, it is convenient to use dimensionless variables and 
parameters\footnote{Following Refs. \cite{naber,iom2011}, one
introduces the Planck length $L_P=\sqrt{\hbar G/c^3}$, time
$T_P=\sqrt{\hbar G/c^5}$, mass $M_P=\sqrt{\hbar c/G}$, and energy
$E_P=M_Pc^2$, where $\hbar,~G,$ and $c$ are the Planck constant,
the gravitational constant and the speed of light, respectively.
Therefore, quantum mechanics of a particle with 
dimensionless mass $m/M_P\rightarrow m$ is
described by the dimensionless units of the coordinates and time
 $x/L_P\rightarrow x,~ t/T_P\rightarrow t$, while the
dimensionless Planck constant is $\hbar/(E_PT_P)\rightarrow 1$.
We however keep the notation of $\hbar$.
 Note, that the dimensionless
frequency is $\hbar\omega/E_P\rightarrow \omega$, which 
keeps $\omega t$ being the dimensionless parameter as well.}
\cite{naber}.

The structure of the paper is as follows. In Sec. \ref{fFt-P} we formulate the fFT with its proof. Then 
the result on the fFT is presented in the form of subordination to the Floquet theorem, in Sec. \ref{com4}.
In Sec \ref{fFt-D} we consider both the Floquet theorem and the fFT in Fourier - Laplace space.
We also consider two 
examples to show how the fFT can be applied for the corresponding FTSEs. The first example, considered in Sec. \ref{toy}, is devoted to a toy model, where the unperturbed spectrum is not affected by time. 
In the second example we consider a quantum particle in a time-dependent potential  and describe it in the framework of the FTSE in Sec. \ref{fqnlr}. The latter situation is also related to quantum nonlinear resonances
\cite{shur76,BeZa77,BeIoZa81}. In Conclusion, we summarize the obtained result on the fFT.
Appendix plays a dual role. First, it is  
an original result on fractional differentiation of periodic functions. At the same time, it is also 
a  brief overview of fractional calculus.

\section{Fractional Floquet theorem}\label{fFt-P}

We formulate the  fractional Floquet theorem  in the form of the following 
proposition.

\begin{prop}

Let us consider the FTSE \eqref{int3}
\[\partial_t^{\alpha}\Psi(t) =\hat{H}(t)\Psi(t) \]
with the time-dependent periodic Hamiltonian $\hat{H}(t+T)=\hat{H}(t)$.
The solution to the FTSE \eqref{int3} can be presented as follows
\begin{equation}\label{fFt-Prop}
\Psi(t)=\sum_nE_{\alpha}\left[i\hbar^{1-\alpha}(\omega n-\epsilon)t^{\alpha}\right]C_n,
\end{equation}
where   $E_{\alpha.1}\equiv E_{\alpha}(z)$ is the one-parameter Mittag -Leffler function  \cite{BaEr53},
$\omega=\frac{2\pi}{T}$, and $C_n$ are the coefficients of the Fourier series of $u(t)$, defined for the standard Floquet theorem in Eq. \eqref{int1}.
\end{prop}

\begin{proof} 

Important part of the consideration is the eigenvalue problem for the Floquet operator in Eq. \eqref{int1},
which after the Fourier expansion, can be presented as follows
\begin{equation}\label{pr-1}
\sum_n\left[\hat{H}\left(\frac{i}{\omega}\frac{d}{d n}\right)C_n+
\hbar(\omega n-\epsilon)C_n\right]e^{i\omega n t}=0,
\end{equation}
see Eq. \eqref{fFt-10} and Comments \ref{com}.
Due to the completeness of the Fourier basis $e^{i\omega n t}$, 
one obtains the eigenvalue equation 
\[
 \hat{H}\left(\frac{i}{\omega}\frac{d}{d n}\right)C_n=
\hbar(\epsilon-\omega n)C_n,~\forall n . 
\]
Then we take into account that the Mittag-Leffler function is the eigenfunction of the 
Caputo fractional derivative, 
\[
\partial_t^{\alpha}E_{\alpha}(iat^{\alpha})=iaE_{\alpha}(iat^{\alpha}) ,
\]
where $a=\hbar^{1-\alpha} (\omega n -\epsilon)$, see Comment \ref{com3}-3. 
Thus, the straightforward substitution of the fFT \eqref{fFt-17} in the FTSE \eqref{int3}
yields
\begin{equation}\label{pr-2}
\sum_n\left[\hbar(-\omega n-\epsilon)C_n+\hat{H}(t)C_n\right]
E_{\alpha}\left[i\hbar^{1-\alpha}(\omega n -\epsilon)t^{\alpha}\right] = 0,
\end{equation}
see Comment \ref{com}-2. 
The expression in square brackets is just the expression in 
 Eq. \eqref{pr-1} that reduces the FTSE to the standard Floquet theory in Eq.  \eqref{int1}.

This proves the Proposition 1 for the fFT.

\end{proof}

\section{Subordination to the Floquet theorem}\label{com4}

Equation \eqref{fFt-Prop} can be presented in a subordination form, 
where the Floquet theorem subordinates the fFT. 
Taking into account the Laplace image of the Mittag-Leffler function, we have the chain of transformations as follows \cite{iom2019}
\begin{multline}\label{com4-1}
E_{\alpha}(iat^{\alpha})=\frac{1}{2\pi i}\int_{-i\infty}^{i\infty}\frac{s^{\alpha-1}e^{st}ds}{s^{\alpha}-ia} 
=\mathcal{L}^{-1}\left[\frac{s^{\alpha-1}e^{st}ds}{s^{\alpha}-ia} \right] \\
=\mathcal{L}^{-1}\left[\int_{-\infty}^{\infty}\frac{s^{\alpha-1}e^{st}ds}{s^{\alpha}-iz}\delta(z-a)dz\right]
=\frac{1}{2\pi}\int_{-\infty}^{\infty}e^{-ia\xi}
\left[\int_{-\infty}^{\infty} e^{iz\xi}E_{\alpha}(izt^{\alpha})\right]d\xi  \\
=\int_{-\infty}^{\infty}\mathcal{K}(\xi,t)\exp\left[i\xi(\omega n -\epsilon)\right]d\xi ,
\end{multline}
where the Fourier transformation of the Mittag-Leffler function, 
\begin{equation}\label{com4-12}
\mathcal{K}(\xi,t)=
\frac{1}{2\pi}\int_{-\infty}^{\infty}e^{-iz\xi}E_{\alpha}(-izt^{\alpha})dz 
\end{equation}
can be reduced to the superposition of the odd and even functions of $z$ and $\xi$. 
Thus, we obtain
\begin{multline}\label{com4-2}
E_{\alpha}(-izt^{\alpha})=\sum_{n=0}^{\infty}\frac{(-z^2t^{2\alpha})^n}{\Gamma(n2\alpha+1)}
-izt^{\alpha}\sum_{n=0}^{\infty}\frac{(-z^2t^{2\alpha})^n}{\Gamma(n2\alpha+\alpha+1)} \\
=E_{2\alpha}(-z^2t^{2\alpha})-izt^{\alpha}E_{2\alpha, 1+\alpha}(-z^2t^{2\alpha}) .
\end{multline}
The first term here is the even function of $z$, which stands for
the cosine-Fourier transformation and correspondingly leads to the
even function of $\xi$. The second term is the odd function of $z$ and
stands for the sine-Fourier transform that leads to the odd function
of $\xi$. Therefore, integration with respect to $\xi\in  (-\infty ,\infty)$ is reduced to the integration 
with respect to $\xi\in(0,\infty)$. Correspondingly, $\xi$ is
treated as a time parameter. Eventually, taking an analogy with a continuous
time random walk (CTRW) theory \cite{SoKl2005}, we suggest that $\mathcal{K}(\xi,t)$
is the subordinator \cite{bochner,nelson}, which subordinates the fractional evolution on the time scale $t$
 (physical time) to the unitary evolution on the time scale $\xi$ (operational time). 
In other words, the fFT is subordinated to the Floquet theorem:
\begin{equation} \label{com4-3}
\Psi(t)=\sum_n\int_0^{\infty}\mathcal{K}(\xi,t)e^{-i\epsilon\xi}e^{i\omega n\xi}C_nd\xi=
\int_0^{\infty}\mathcal{K}(\xi,t)e^{-i\epsilon\xi}u(\xi).
\end{equation}

\section{Discussion: Floquet theorem in Fourier-Laplace space}\label{fFt-D}

We start the discussion with the case of $\alpha=1$,  when
$E_{\alpha}\left[i\hbar^{1-\alpha}(\omega n-\epsilon)t^{\alpha}\right]=e^{i(\omega n-\epsilon)t}$.
Then Eq. \eqref{fFt-Prop} reduces to
the standard Floquet theorem,
which states that the solution to the Schr\"odinger equation 
\begin{equation}\label{fFt-1}
i\hbar\partial_t\psi(t)=\hat{H}\psi(t), \quad \quad \hat{H}(t+T)=\hat{H}(t)
\end{equation}
is 
\begin{equation}\label{fFt-2}
\psi(t)=\sum_ne^{i(\omega n-\epsilon)t}C_n =e^{-i\epsilon t}u(t), 
\end{equation}
where $u(t)$ is the periodic function, which together with the Hamiltonian  can be presented in 
the form of the the Fourier expansion
\begin{subequations}\label{fFt-3}
\begin{align}
& u(t)= \sum_{n=-\infty}^{\infty}C_ne^{in\omega t}, \quad \quad \omega=\frac{2\pi}{T},  \label{fFt-3a} \\
& \hat{H}(t)=\sum_{m=-\infty}^{\infty}\hat{h}(m)e^{im\omega t}  .      \label{fFt-3b}
\end{align}
\end{subequations}

The Laplace transform of Eqs. \eqref{fFt-2} and \eqref{fFt-3a} yields
\begin{equation}\label{fFt-4}
\mathcal{L}\left[\psi(t)\right]=\tilde{\psi}(s)=\sum_nC_n\left(s+i\epsilon-i\omega n\right)^{-1},
\end{equation}
and the Laplace transform of the Schr\"odinger Eq. \eqref{fFt-1} yields
\begin{equation}\label{fFt-5}
i\hbar s \sum_n\frac{C_n}{s+i\epsilon-i\omega n}=
\sum_n\left(\sum_m\frac{\hat{h}(m)C_{n-m}}{s+i\epsilon -i\omega n} +i\hbar C_n\right),
\end{equation}
where we take into account that the initial condition is $\Psi_0=\psi_0=\psi(t=0)=\sum_nC_n$, according to 
the Floquet theorem \eqref{fFt-2} and Eq. \eqref{fFt-3a}.

Considering $C_n$ as a ``lattice amplitude'', we have 
\begin{equation}\label{fFt-6}
C_n=\sum_qU_qe^{iqn},
\end{equation}
and substituting the series of Eq. \eqref{fFt-6} in Eq. \eqref{fFt-5}, we obtain,
\begin{equation}\label{fFt-7}
\sum_q\sum_n\left[\frac{1}{s+i\epsilon-i\omega n}\left(i\hbar s-\sum_m\hat{h}(m)e^{-iqm}\right)
-i\hbar\right]U_qe^{iqn}=0.
\end{equation}
According to Eq. \eqref{fFt-3b}, summation with respect to $m$ yields 
$\sum_m\hat{h}(m)e^{-iqm}=\hat{H}(-q/\omega)$. Then taking into account
Eq. \eqref{fFt-6} and the fact that 
\begin{equation}\label{fFt-8}
\hat{H}\left(\frac{i}{\omega}\frac{d}{d n}\right)e^{iqn}=
\hat{H}\left(-\frac{q}{\omega}\right)e^{iqn},
\end{equation}
we obtain Eq. \eqref{fFt-7} as follows
\begin{equation}\label{fFt-9}
\sum_n\frac{1}{s+i\epsilon-i\omega n}
\left[\hat{H}\left(\frac{i}{\omega}\frac{d}{d n}\right)C_n
-(\hbar\epsilon-\hbar\omega n)C_n\right] =0.
\end{equation}
Therefore for any values of $\epsilon$ and $n$, the coefficients $C_n$ are eigenfunctions of the Hermitian
operator $\hat{H}\left(\frac{i}{\omega}\frac{d}{d n}\right)$, namely
\begin{equation}\label{fFt-10}
\hat{H}\left(\frac{i}{\omega}\frac{d}{d n}\right)C_n=(\hbar\epsilon-\hbar\omega n)C_n .
\end{equation}
Note that the Laplace argument $s$ is not restricted as well, since 
performing the Laplace inversion, the denominator in Eq. \eqref{fFt-9} becomes just the exponential
function, $e^{-it(\epsilon-\omega n)}$.

\subsection{Floquet theorem for the FTSE}

Let us return to the FTSE  \eqref{int3}, and consider its solution in the form
\begin{equation}\label{fFt-11}
\Psi(t)=\sum_n\Lambda_n(t)C_ne^{i\omega nt},
\end{equation}
where $\Lambda_n(t)$ is an unknown function of time.
The Laplace transformations of both the FTSE  \eqref{int3} and the solution \eqref{fFt-11} are
\begin{subequations}\label{fFt-12}
\begin{align}
i\hbar^{\alpha}s^{\alpha}\tilde{\Psi}(s)=\sum_m\hat{h}(m)\sum_nC_n
\mathcal{L}\left[\Lambda_n(t)e^{it\omega(n+m)}\right] +i\hbar^{\alpha}s^{\alpha-1}\Psi_0 
\nonumber \\
=\sum_m\sum_n\hat{h}(m)C_n\tilde{\Lambda}(s-i\omega n-i\omega m)
+i\hbar^{\alpha}s^{\alpha-1}\Psi_0 , \label{fFt-12a} \\
\tilde{\Psi}(s)=\sum_n\tilde{\Lambda}(s-i\omega n) C_n, \quad \quad 
\Psi_0=\Lambda(0)\sum_nC_n .  \label{fFt-12b} 
\end{align}
\end{subequations}
Performing the shift of indexes $\sum_n\tilde{\Lambda}(s-i\omega n-i\omega m)C_n=
\sum_n\tilde{\Lambda}(s-i\omega n)C_{n-m}$, we obtain the equation for $\tilde{\Lambda}(s-i\omega n)$
 from Eqs. \eqref{fFt-12a} and \eqref{fFt-12b} as follows
\begin{equation}\label{fFt-13}
\sum_n\left[\tilde{\Lambda}(s-i\omega n)\left( i\hbar^{\alpha}s^{\alpha}C_n-
\sum_m\hat{h}(m)C_{n-m}\right)-i\hbar^{\alpha}s^{\alpha-1}\Lambda(0)C_n\right]=0.
\end{equation}

One should bear in mind that the coefficients $C_n$ and the Hamiltonian $\hat{H}(t)$
are the same as in the standard Floquet theory. Therefore, performing the expansion 
\eqref{fFt-6} for the lattice amplitudes $C_n$ and then performing summation
with respect to $m$, we obtain $\sum_m\hat{h}(m)e^{-iqm}=\hat{H}(-q/\omega)$.
Then taking into account 
Eq. \eqref{fFt-8} and the eigenvalue Eq. \eqref{fFt-10}, we obtain  Eq.
\eqref{fFt-13} as follows
\begin{multline}\label{fFt-14}
\sum_n\sum_q\left\{\tilde{\Lambda}(s-i\omega n)
\left[i\hbar^{\alpha}s^{\alpha}-\hat{H}(-q/\omega)\right]
-i\hbar^{\alpha}s^{\alpha-1}\Lambda(0)\right\}U_qe^{iqn} \\
=\sum_n\left\{\tilde{\Lambda}(s-i\omega n)
\left[i\hbar^{\alpha}s^{\alpha}-\hat{H}\left(\frac{i}{\omega}\frac{d}{dn}\right)\right]
-i\hbar^{\alpha}s^{\alpha-1}\Lambda(0)\right\}C_n \\
=\sum_n\left\{\tilde{\Lambda}(s-i\omega n)
\left[i\hbar^{\alpha}s^{\alpha}-\hbar(\epsilon-\omega n)\right]
-i\hbar^{\alpha}s^{\alpha-1}\Lambda(0)\right\}C_n =0.
\end{multline}
Since the coefficients $C_n$ are the functions in question in the Floquet theory, it is reasonable to 
suppose that the Laplace image is 
\begin{equation}\label{fFt-15}
\tilde{\Lambda}(s-i\omega n)=
\frac{s^{\alpha-1}\Lambda(0)}{s^{\alpha}+i\hbar^{1-\alpha}(\epsilon-\omega n)} ,
\end{equation}
which is the Laplace image of the Mittag-Leffler function
\cite{BaEr53}. Performing the Laplace inversion and setting $\Lambda(0)=1$, we obtain
the solution in the form of the one parameter Mittag-Leffler function,
\begin{equation}\label{fFt-16}
\Lambda_n(t)e^{i\omega nt}=E_{\alpha}\left[i\hbar^{1-\alpha}(\omega n-\epsilon)t^{\alpha}\right].
\end{equation}
Taking into account the definition \eqref{fFt-11}, we arrived at the wave function as follows
\begin{equation}\label{fFt-17}
\Psi(t)=\sum_n\Lambda_n(t)e^{i\omega nt}C_n=
\sum_nE_{\alpha}\left[i\hbar^{1-\alpha}(\omega n-\epsilon)t^{\alpha}\right]C_n.
\end{equation}
We stress that it has been proven in Sec. \ref{fFt-P} that the wave function in the form of \eqref{fFt-17} is the fFT. 

 \subsection{Comments}\label{com}
 
\begin{itemize}
\item[1]{Evolution operator}\label{com1}

The evolution of the wave function of the FTSE \eqref{int3} is according to the evolution 
operator $\hat{U}(t)$ (do not confuse it with the band amplitude $U_q$), such that 
$\Psi(t)=\hat{U}(t)\Psi_0$. However, the explicit operator form of the evolution operator is unknown,
since the fractional Caputo derivative destroys the periodicity of the wave function, see \ref{app-A}.
Therefore, the Fourier expansion in Eqs. \eqref{fFt-3a} and \eqref{fFt-11} is formal, such that 
the coefficient $C_n$ can be both an operator valued function and a function of coordinate space,
depending on the explicit operator form of the Hamiltonian $\hat{H}(t)$. The former case is considered 
as an example in Sec. \ref{toy}, while the latter case is an example of Sec. \ref{fqnlr}.

\item[2] {Lattice amplitude}\label{com2}

Important part of the analysis is relates to the consideration of the amplitude $C_n$ as a ``lattice amplitude''  
in the form of the expansion \eqref{fFt-6}, namely
$C_n=\sum_qU_qe^{iqn}$. Then taking into account the expansion  \eqref{fFt-3b} and performing summation
with respect to $m$, one obtains that the time dependence in the Hamiltonian can be replaced as follows
 $\sum_m\hat{h}(m)e^{-iqm}=\hat{H}(-q/\omega)$.  The latter leads to expression \eqref{fFt-8}, which is
\[
\hat{H}\left(\frac{i}{\omega}\frac{d}{d n}\right)e^{iqn}=
\hat{H}\left(-\frac{q}{\omega}\right)e^{iqn}.
\]
Therefore, taking into account this chain of transformations (that always can be done),
one can made the formal replacement 
\[
\hat{H}(t)C_n\equiv \hat{H}(\omega t)C_n\rightarrow \hat{H}\left(i\frac{d}{d n}\right)C_n.
\]

\item[3] {Mittag-Leffler function}\label{com3}

Note, that the Mittag-Leffler function is the eigenfunction of the Caputo fractional derivative. 
Indeed, taking into account Eq. \eqref{ch2_1-RL-7}, we obtain
\begin{multline*}
\partial_t^{\alpha}E_{\alpha}(iat^{\alpha})=\sum_{k=0}^{\infty}\frac{(ia)^k}{\Gamma(k\alpha+1)}
\partial_t^{\alpha}t^{k\alpha} \\
=\sum_{k=1}^{\infty}\frac{(ia)^k}{\Gamma(k\alpha+1)}\cdot
\frac{\Gamma(k\alpha+1)}{\Gamma(k\alpha-\alpha+1)}t^{(k-1)\alpha}  \\
=ia\sum_{k=0}^{\infty}\frac{(ia)^kt^{k\alpha}}{\Gamma(k\alpha+1)}=iaE_{\alpha}(iat^{\alpha}) ,
\end{multline*}
where $a=\hbar^{1-\alpha} (\omega n -\epsilon)$.

\end{itemize}

\section{Example 1: Toy model}\label{toy}

In this section, we consider a toy Hamiltonian presented in the form
\begin{equation}\label{toy-1}
\hat{H}(t)=\hat{H}_0\cos(\omega t),
\end{equation}
where $\hat{H}_0$ can be an arbitrary time-independent Hamiltonian, which can
describe a variety of conservative systems. The standard Schr\"odinger equation for the toy model \eqref{toy-1}
can be easily solved, and the evolution of the initial wave function $\psi_0$ is
\begin{equation}\label{toy-2}
\psi(t)=\exp\left(-\frac{i}{\hbar}\int_0^t\hat{H}(t')dt'\right)\psi_0 =
\exp\left[-i\frac{\hat{H}_0}{\hbar\omega}\sin(\omega t)\right]\psi_0 .
\end{equation}
The initial wave function $\psi_0$ can be \textit{e.g.}, the eigenfunction of $\hat{H_0}$, namely $\hat{H}_0\psi_0=e\psi_0$.
Then $\psi(t)=\exp[-i(e/\hbar\omega)\sin(\omega t)]$. Since for $t=kT$ ($k$ is integer), $\psi(kT)=\psi_0$,
correspondingly the quasienergy is zero, $\epsilon=0$ and $\psi(t)=u(t)$.
The Fourier transformation yields
\begin{equation}\label{toy-3}
\psi(t)=\sum_{n=-\infty}^{\infty}J_{-n}\left(\frac{\hat{H}_0}{\hbar\omega}\right)e^{i\omega nt},
\end{equation}
where $C_n=(-1)^nJ_n(z)\equiv J_{-n}\left(\frac{\hat{H}_0}{\hbar\omega}\right) $ is the Bessel function of the first kind, which also satisfies the eigenvalue Eq. \eqref{fFt-10}. The latter reads 
\begin{equation}\label{toy-4}
z\cos\left(\frac{d}{dn}\right)J_{-n}(z)= -n J_{-n}(z) ,
\end{equation}
which is  just the recurrence relation of the Bessel functions \cite{AbSt72}: 
\[
zJ_{n+1}(z)+zJ_{n-1}(z)=2nJ_n(z) .
\]

\subsection{Fractional toy model}\label{ftoy}
The fractional generalization of the toy model is formulated in the framework of the FTSE
as follows
\begin{equation}\label{toy-5}
i\hbar^{\alpha}\partial_t^{\alpha}\Psi(t)=\hat{H}_0\cos(\omega t) \Psi(t)
\end{equation}
where the initial condition is $\Psi_0=\psi_0$.
According to the fFT, we look for the solution to the FTSE \eqref{toy-5} as follows
\begin{equation}\label{toy-6}
\Psi(t)=\sum_nE_{\alpha}\left(i\hbar^{1-\alpha}\omega nt^{\alpha}\right)
(-1)^nJ_n\left(\frac{\hat{H}_0}{\hbar\omega}\right)\psi_0 ,
\end{equation}
which for $\alpha=1$ reduces to Eq. \eqref{toy-2}.

Let us substitute this fFT solution into the FTSE \eqref{toy-4}. Taking into account the 
Comment \ref{com3}-3,
we obtain 
\begin{multline}\label{toy-7}
\sum_n(-\hbar\omega n )E_{\alpha}\left(i\hbar^{1-\alpha}\omega nt^{\alpha}\right)
(-1)^nJ_n\left(\frac{\hat{H}_0}{\hbar\omega}\right)\psi_0 \\
=
\hat{H}_0\cos(\omega t) \sum_nE_{\alpha}\left(i\hbar^{1-\alpha}\omega nt^{\alpha}\right)
(-1)^nJ_n\left(\frac{\hat{H}_0}{\hbar\omega}\right)\psi_0 .
\end{multline}
Performing the chain of transformations according to the Comment 
\ref{com2}-2, we obtain the 
rhs of Eq. \eqref{toy-7} as follows
\[
\hat{H}_0 \sum_nE_{\alpha}\left(i\hbar^{1-\alpha}\omega nt^{\alpha}\right)
\cdot\left[\cos\left(\frac{d}{dn}\right)(-1)^nJ_n\left(\frac{\hat{H}_0}{\hbar\omega}\right)\right]\psi_0 .
\]
Then from Eq. \eqref{toy-4}, we get the identity for Eq. \eqref{toy-7}.

\section{Example 2: Fractional quantum dynamics of a particle in time-dependent potential}\label{fqnlr}

In this section, we consider another, more realistic example, where the coefficients $C_n=C_n(x)$ 
are defined in the coordinate space $x\in R$. It is a quantum particle in a nonlinear time-dependent potential
$V(x)\cos(\omega t)$, 
which can be described \textit{e.g.}, by the Schr\"odinger equation
\begin{equation}\label{fqnlr-1}
i\hbar\partial_t\psi(x,t)=\hat{H}(t)\psi(x,t)\equiv
\left[-\frac{\hbar^2}{2m}\frac{d^2}{dx^2}+V(x)\cos(\omega t)\right]\psi(x,t) ,
\end{equation}
where $m$ is the mass of a quantum particle affected by the time-dependent potential
$V(x)\cos(\omega t)$, where $V(x)$ is a nonlinear function of $x$. 
Such systems can describe a quantum nonlinear resonance  \cite{shur76,BeZa77}. 
We concentrate our attention on the Floquet theory of the Schr\"odinger Eq. \eqref{fqnlr-1} and shall not discuss the physics of the quantum resonance, which is well studied and  well reviewed in literature \cite{Zas85,BerKa2001,Reichl2004}.
Therefore, according to the Floquet theorem,  the eigenvalue problem for the Floquet operator reads
\begin{equation}\label{fqnlr-2}
\hat{H}(t)u(x,t)-i\hbar\partial_tu(x,t)=\hbar\epsilon u(x,t),
\end{equation}
where $u(x,t)=\sum_nC_n(x)e^{i\omega nt}$. This yields the equation for the coefficients $C_n(x)$ as follows
\begin{multline}\label{fqnlr-3}
\hat{K}C_n(x)\equiv \\
-\frac{\hbar^2}{2m}\frac{d^2}{dx^2}C_n(x)+\frac{1}{2}V(x)\left[C_{n+1}(x)+C_{n-1}(x)\right]+
\hbar(\omega n-\epsilon)C_n(x)=0.
\end{multline}
After diagonalization, when $C_n(x)=\sum_qU_q(x)e^{iqn}$, Eq. \eqref{fqnlr-3} reads
\begin{equation}\label{fqnlr-4}
-\frac{\hbar^2}{2m}\frac{d^2}{dx^2}U_q(x)+V(x)\cos(qn)U_q(x)+\hbar\omega nU_q(x)=\hbar\epsilon U_q(x).
\end{equation}
Depending on the explicit form of the potential $V(x)$, this equation can be treated either
analytically, or numerically \cite{Zas85,BerKa2001,Reichl2004}. We however, do not concern with the issue,
as admitted above.

Equation \eqref{fqnlr-3} is the skeleton equation  for the consideration of  the  fractional quantum nonlinear resonance in the framework of the FTSE \eqref{fqnlr-1}, which reads now
\begin{equation}\label{fqnlr-5}
i\hbar\partial_t^{\alpha}\Psi(x,t)=\hat{H}(t)\Psi(x,t)\equiv
\left[-\frac{\hbar^2}{2m}\frac{d^2}{dx^2}+V(x)\cos(\omega t)\right]\Psi(x,t) ,
\end{equation}

To obtain Eq.  \eqref{fqnlr-3}, let us substitute  expression \eqref{fFt-Prop} into FTSE \eqref{fqnlr-5}. 
Following the Comment \ref{com3}-2, we obtain
\begin{multline}\label{fqnlr-6}
\sum_n\hbar(\epsilon -\omega n)E_{\alpha}\left[i\hbar^{1-\alpha}(\epsilon -\omega n)t^{\alpha}\right]C_n(x) \\
= \sum_n E_{\alpha}\left[i\hbar^{1-\alpha}(\epsilon -\omega n)t^{\alpha}\right]
\left[-\frac{\hbar^2}{2m}\frac{d^2C_n(x)}{dx^2}+V(x)\cos(\omega t)C_n \right] .
\end{multline}
Following the Comment \ref{com2}, we have 
$\cos(\omega t)C_n \rightarrow \frac{1}{2}(C_{n+1}+C_{n-1}) $. Therefore, in the Laplace space,
the real and imaginary parts of  Eq. \eqref{fqnlr-6} are
\begin{subequations}\label{fqnlr-7}
\begin{align}
\sum_n\frac{s^{2\alpha-1}}
{s^{2\alpha}+\hbar^{2-2\alpha}(\epsilon -\omega n)^2}\left[\hat{K}C_n(x)\right]=0, \\
i\hbar^{1-\alpha}\sum_n\frac{(\epsilon -\omega n)}
{s^{2\alpha}+\hbar^{2-2\alpha}(\epsilon -\omega n)^2}\left[\hat{K}C_n(x)\right]=0.
\end{align}
\end{subequations}
Then we arrive at Eq. \eqref{fqnlr-3}: $\hat{K}C_n(x)=0$.

\section{Conclusion}\label{concl}
In the paper we concern with the fractional Schr\"odinger equations (FTSE) for systems, which are described by the time-dependent periodic Hamiltonians. An analog of the well known Floquet theorem  is suggested to solve this class of the FTSEs,  and this new result  is called the fractional Floquet theorem (fFT).
The main result can be formulated in the form of the following generalization of the Floquet theorem for the wave function. Namely, considering the standard  Floquet theorem  as follows 
\begin{equation}\label{concl-1}
\psi(t)=e^{-i\epsilon t}u(t) =\sum_nC_ne^{-i\epsilon t+i\omega nt} ,
\end{equation}
we generalize the exponential by the one-parameter Mittag-Leffler function 
$E_{\alpha}\left[i\hbar^{1-\alpha}(\omega n -\epsilon)t^{\alpha}\right]$, which yields
the fFT \eqref{fFt-Prop} as follows
\begin{equation}\label{concl-2}
\Psi(t)=\sum_nC_nE_{\alpha}\left[i\hbar^{1-\alpha}(\omega n -\epsilon)t^{\alpha}\right],
\end{equation}
which for $\alpha=1$ reduces to the Floquet theorem \eqref{concl-1}.
The fFT has been verified by the straightforward substitution of Eq. \eqref{fFt-Prop} in the FTSE that 
yields the identity. 
Two examples supporting the obtained result have been considered as well. 

In conclusion, it should be admitted that by means of the Mittag-Leffler function the fFT reduces the FTSE \eqref{int3} with the time dependent Hamiltonian to the standard Floquet theorem consideration. In this case, the  quasienergy spectrum of the FTSE exists and is determined by the standard Floquet theorem.


\section*{Acknowledgments} 
It is my pleasure to acknowledge the hospitality at the Max Planck Institute for the Physics of Complex Systems, Dresden, where a part of the work has been done.

\begin{appendix}

\section{Fractional differentiation of periodic functions}
\label{app-A}

In this notes, we consider fractional calculus as an example of fractional differentiation of periodic functions. Recently, this issue has attracted 
some attention in the form of non-existence of periodic solutions in fractional-order dynamical systems \cite{Kaslik2012}. 
It has been shown in the framework of the Mellin transform consideration that fractional differentiation of periodic functions destroys their periodicity.  The result was obtained for the Caputo, Riemann-Liouville and 
Gr\"{u}nwald-Letnikov definitions of fractional-order derivatives.
We suggest an alternative approach, following the fundamental work of Ref. 
\cite{SaKiMa93}, and the Fourier and the Laplace transformations are the main machinery of the analysis.

Any periodic function $g(t)$ of the period $T$, such that $g(t+T)=g(t)$,
can be presented in the form of its Fourier series
\begin{equation}\label{ch2_1-RL-1}
g(t)=\sum_{l=-\infty}^{\infty}g_l e^{i\bar{l}t}\, , \quad
g_l=\frac{1}{T}\int_0^Tg(t)e^{-i\bar{l}t}dt\, ,
\quad \bar{l}=2\pi l/T\, .
\end{equation}
Therefore, $g(t)$ is determined by its Fourier image $g_l$ if the Fourier series \eqref{ch2_1-RL-1} converges. In particular, the $n$-th derivative is a periodic function, determined by $\bar{l}^ng_l$ as the convergence of 
$g^{(n)}(t)=\tfrac{d^ng(t)}{dt^n}$ in Eq. \eqref{ch2_1-RL-1}. 
Therefore, to define periodicity or non-periodicity of the fractional 
integro-differentiation of $g(t)$ one considers the exponential
$\exp[2\pi i lt/T]$ only.

\subsection{Riemann-Liouville fractional integral}
\label{ch2_1-RL}

Note that the Caputo and Riemann-Liouville  
fractional derivatives are regularization of the fractional derivative 
\begin{equation}\label{ch2_1-RL-2}
\fourIdx{}{a+}{\mu}{t}{D}f(t)
=\frac{1}{\Gamma(\mu)}
\int_{a}^{t}(t-\tau)^{-\mu-1}f(\tau)\,d\tau,
\end{equation}
where we take $a=0$.
Therefore, taking $n-1<\mu<n$ and $\nu=n-\mu$, and dropping out the integer part of the derivative from the consideration, 
we arrive at the Riemann-Liouville  fractional integration of the exponential
$f(t)=e^{i\bar{l}\tau}$,
\begin{equation}\label{ch2_1-RL-3}
\fourIdx{}{0}{-\nu}{t}{D}f(t)\equiv
I_{0+}^{\nu}f(t) =
\frac{1}{\Gamma(\nu)}
\int_{0}^{t}(t-\tau)^{\nu-1}e^{i\bar{l}\tau}\,d\tau.
\end{equation}
To treat this integral, we use the Laplace transform of the convolution integral, which we present it in the form of the identity
\begin{equation}\label{ch2_1-RL-4}
I_{0+}^{\nu}e^{i\bar{l}t}=\mathcal{L}^{-1}
\left\{\mathcal{L}\left[I_{0+}^{\nu}e^{i\bar{l}t}\right]\right\} 
=\frac{1}{2\pi i}\int_C\frac{s^{\nu}e^{st}}{s-i\bar{l}}ds=
t^{-\nu}E_{1,1-\nu}\left(i\bar{l}t\right)\, ,
\end{equation}
where the last line is the integral representation of the two parameter
Mittag-Leffler function $E_{\alpha,\beta}(z)$, \cite{BaEr53}.
Taking into account that $1-\nu=\{\mu\}$ is the fractional part of $\mu$,
and presenting the  Mittag-Leffler  function as a superposition of the real and imaginary parts, we  obtain
\begin{multline}\label{ch2_1-RL-5}
E_{1,\{\mu\}}\left(i\bar{l}t\right)=\sum_{k=0}^{\infty}\left[
\frac{\left(i\bar{l}t\right)^{2k}}{\Gamma(2k+\{\mu\})}
+\frac{\left(i\bar{l}t\right)^{2k+1}}{\Gamma(2k+1+\{\mu\})} \right]\\
=E_{2,\{\mu\}}\left(-\bar{l}^2t^2\right)+
i\bar{l}tE_{2,1+\{\mu\}}\left(-\bar{l}^2t^2\right)\, .
\end{multline}
Eventually, we show that the Riemann-Liouville fractional integration destroys the periodicity of periodic functions. Therefore, both the Caputo and Riemann-Liouville fractional derivatives act in the same way.

\subsection{Fox $H$-function} 

Let us present exponential in the form of the Fox $H$-function,
\[
e^{-z}=H_{0,1}^{1,0}\left[z\Big\vert\binom{}{(0,1)}\right].
\]
Thus the exponential can be presented in the form of the Mellin-Barnes
integral \cite{MaSaHa10}
\begin{equation}\label{ch2_1-RL-6}
e^{i\bar{l}t}=\frac{1}{2\pi i}\int_C
\Gamma(\xi)(-i\bar{l}t)^{-\xi}d\xi
\end{equation}
Making the variable change $\xi\rightarrow -\xi$, we obtain 
the singularities of the gamma functions $\Gamma(-\xi)$ at
$\Re{\xi}\ge 0$. (In this case, Eq. \eqref{ch2_1-RL-6}
corresponds to the Meijer $G$-function \cite{BaEr53}.)
Now let us consider the Caputo fractional derivative $\partial^{\{\mu\}}_t$
of the power law function $t^{\xi}$, which yields
\begin{multline}\label{ch2_1-RL-7}
\partial^{\{\mu\}}_t(-i\bar{l}t)^{\xi}= (-i\bar{l})^{\xi}\frac{\Gamma(\xi+1)}{\Gamma(1-\{\mu\})}
\int_0^t\tau^{\xi-1}(t-\tau)^{-\{\mu\}}d\tau \\
=(-i\bar{l})^{\xi} 
\frac{\Gamma(\xi+1)}{\Gamma(\xi-\{\mu\}+1)}t^{\xi-\{\mu\}}\, ,
\end{multline}
where $\{\mu\}=\mu+1-n$.
Combining this result with Eq. \eqref{ch2_1-RL-6}, 
then performing the variable change $\xi\rightarrow -\xi$  with corresponding deformation of the countour,
we obtain the Mellin-Barnes integral for the Mittag-Leffler function
\begin{equation}\label{ch2_1-RL-8}
\partial^{\{\mu\}}_te^{i\bar{l}t}=\frac{t^{-\{\mu\}}}{2\pi i}\int_C
\frac{\Gamma(\xi)\Gamma(1-\xi)}{\Gamma(1-\mu-\xi)}(-i\bar{l}t)^{-\xi}d\xi
=t^{-\{\mu\}}E_{1,1-\{\mu\}}\left(i\bar{l}t\right)
\end{equation}

\subsection{Riesz fractional derivatives}\label{ch2_1-RF}

In this section we treat the periodic function $g(x)=g(x+L)$, 
where $x\in R$ in the framework of the symmetric Riesz-Feller integration.
Let us consider the symmetric Riesz-Feller derivative, which can be defined  in Fourier space 
\begin{equation}\label{ch2_1-RF-1}
\mathcal{F}\left[\fourIdx{\textrm{RF}}{0}{\mu}{x}{D}f(x)\right](k)=
-|k|^{\mu}\mathcal{F}\left[f(x)\right](k).
\end{equation}
Its integral representation reads 
\begin{multline}\label{ch2_1-RF-2}
    \fourIdx{\textrm{RF}}{0}{\mu}{x}{D}f(x)=\frac{\Gamma(1+\mu)}{\pi}\sin\frac{\mu\pi}{2}\int_{0+}^{\infty}\frac{f(x+\xi)-2f(x)+f(x-\xi)}
    {\xi^{1+\mu}}d\xi \\
    = \frac{\mu}{2\Gamma(1-\mu)\cos\frac{\mu\pi}{2}}
    \int_{0}^{\infty}\frac{f(x+\xi)-2f(x)+f(x-\xi)}
    {\xi^{1+\mu}}d\xi
\end{multline}
Setting $f(x)=e^{i\bar{l}x}$ with $\bar{l}=2\pi l/L$, integration in Eq. \eqref{ch2_1-RF-2} yields
\begin{multline}\label{ch2_1-RF-3}
    \fourIdx{\textrm{RF}}{0}{\mu}{x}{D}f(x)
    = \frac{\mu e^{i\bar{l}x}}{2\Gamma(1-\mu)\cos\frac{\mu\pi}{2}}
    \int_{0}^{\infty}\frac{\sin^2(\bar{l}\xi)}{\xi^{1+\mu}}d\xi \\
    = -\frac{2e^{i\bar{l}x}}{\Gamma(1-\mu)\cos\frac{\mu\pi}{2}}
    \int_{0}^{\infty}\sin^2(\bar{l}\xi)d(\xi^{-\mu}) =
    A(\mu)e^{i\bar{l}x}, \quad 0<\mu<2\, ,
\end{multline}
where $A(\mu)=(2\bar{l})^{\mu}$. Therefore, the symmetric
Riesz-Feller differentiation of periodic functions does not destroy the  periodicity, namely 
$\fourIdx{\textrm{RF}}{0}{\mu}{x}{D}g(x)$ is the periodic
function with the same period $L$ of $g(x)$ for $0<\mu<2$.

\subsection{Riesz fractional derivative in Fourier space}

Let us consider left $(+)$ and right $(-)$ fractional integrals
\begin{subequations}\label{ch2_1-RF-4}
\begin{align}
&(I_{+}^{\mu}f)(x)=\frac{1}{\Gamma(\mu)}\int_{-\infty}^{x}
(x-y)^{\mu-1}f(y)dy,  \label{ch2_1-RF-4a} \\
&(I_{-}^{\mu}f)(x)=\frac{1}{\Gamma(\mu)}\int^{\infty}_{x}
(y-x)^{\mu-1}f(y)dy . \label{ch2_1-RF-4b}
\end{align}
\end{subequations}
Correspondingly the R-L fractional derivatives are 
\cite{SaKiMa93}(\S ~5.1)
\begin{subequations}\label{ch2_1-RF-5}
\begin{align}
&(D_{+}^{\mu}f)(x)=\frac{1}{\Gamma(n-\mu)}\frac{d^n}{d x^n}\int_{-\infty}^{x}
(x-y)^{n-\mu-1}f(y)dy,  \label{ch2_1-RF-5a} \\
&(D_{-}^{\mu}f)(x)=\frac{(-1)^n}{\Gamma(n-\mu)}\frac{d^n}{d x^n}
\int^{\infty}_{x} (y-x)^{n-\mu-1}f(y)dy , \label{ch2_1-RF-5b}
\end{align}
\end{subequations}
where $n-1<\mu<n$.
Then the Fourier transform of Eq. \eqref{ch2_1-RF-5} reads
\cite{SaKiMa93}(\S ~7.1)
\begin{equation}\label{ch2_1-RF-6}
\mathcal{F}\left[(D_{\pm}^{\mu}f)(x)\right](k)=
(\mp i k)^{\mu}\tilde{f}(k)\,  .
\end{equation}

Introducing the Riesz fractional derivative, we consider
superposition of the Fourier images in
Eq. \eqref{ch2_1-RF-6}, which is
\begin{multline}\label{ch2_1-RF-7}
\mathcal{F}\left[(D_{-}^{\mu}f)(x)\right](k)+
\mathcal{F}\left[(D_{+}^{\mu}f)(x)\right](k) \\
= \left[(i k)^{\mu}+(-i k)^{\mu}\right]\tilde{f}(k)
=2|k|^{\mu}\cos\frac{\pi\mu}{2}\tilde{f}(k)\, .
\end{multline}
Applying this result to the periodic function $f(x)=e^{i\bar{l}x}$,
we obtain $\tilde{f}(k)=2\pi\delta(k-\bar{l})$ that ensures 
the fractional operator keeps the periodicity of $g(x)$, which is 
immediate by the Fourier inversion.

By performing the Fourier inversion of Eq. \eqref{ch2_1-RF-7}, the Riesz fractional derivative $(-\Delta)^{\frac{\mu}{2}}$ can be defined. In the one dimensional case it is
\begin{multline}\label{ch2_1-RF-8}
(-\partial^2_x)^{\frac{\mu}{2}}f(x)\equiv\partial^{\mu}_{|x|}f(x)=
\frac{1}{2\cos\frac{\pi\mu}{2}}\left[D_{+}^{\mu}+D_{-}^{\mu}\right] \\
=\int_{-\infty}^{\infty}dyf(y)\frac{1}{2\pi}\int_{-\infty}^{\infty}|k|^{\mu}
e^{-i k(x-y)}dk\, ,
\end{multline}
where $\mathcal{F}^{-1}[|k|^{\mu}](x-y)$ is the kernel.
Let us obtain it for $\mu\in(0\, , 2)$. Thus we have
\begin{multline}\label{ch2_1-RF-9}
\frac{1}{2\pi}\int_{-\infty}^{\infty}|k|^{\mu}e^{-i k(x-y)}dk=
\frac{1}{\pi}\int_0^{\infty}k^{\mu}\cos[k(x-y)]dk \\
=\frac{1}{\pi}\left(-\frac{d^2}{dx^2}\right)
\int_0^{\infty}k^{\mu-2}\cos[k(x-y)]dk \\
=
\left(-\frac{d^2}{dx^2}\right)
\frac{(x-y)^{1-\mu}}{2\Gamma(2-\mu)\cos\left[\frac{(2-\mu)\pi}{2}\right]}
=\frac{(x-y)^{-\mu-1}}{2\Gamma(-\mu)\cos\frac{\mu\pi}{2}} \, .
\end{multline}
Eventually, we obtain
\begin{multline}\label{ch2_1-RF-10}
\partial^{\mu}_{|x|}f(x)=
\frac{1}{2\Gamma(-\mu)\cos\frac{\mu\pi}{2}}\int_{-\infty}^{\infty}|x-y|^{-\mu-1}f(y)dy \\
=
\frac{1}{2\Gamma(-\mu)\cos\frac{\mu\pi}{2}}\int_{-\infty}^{\infty}
f(x-y)|y|^{-\mu-1}\, ,
\end{multline}
which for $\mu>0$ diverges and demands regularisation.
As shown here, one of the regularization for $\mu\in(0\, , 2)$
is 
\begin{equation}\label{ch2_1-RF-11}
\partial^{\mu}_{|x|}f(x)=
\frac{1}{2\Gamma(-\mu)\cos\frac{\mu\pi}{2}}\int_{0+}^{\infty}
y^{-\mu-1}\left[f(x+y)-2f(x)+f(x-y)\right]dy\, .
\end{equation}
Taking into account that $\Gamma(1-\mu)=-\mu\Gamma(-\mu)$, Eq.
\eqref{ch2_1-RF-11} can be rewritten in the form of the 
Gr{\"u}nwald-Letnikov-Riesz fractional derivative \cite{SaKiMa93}
\begin{equation}\label{ch2_1-RF-12}
f^{\mu)}(x)=
\frac{\mu}{2\Gamma(1-\mu)\cos\frac{\mu\pi}{2}}\int_{0+}^{\infty}
\frac{2f(x)-f(x+y)-f(x-y)}{y^{\mu+1}}dy\, .
\end{equation}

\subsection{Gr{\"u}nwald-Letnikov derivative}\label{ch2_1-GL}

The Gr{\"u}nwald-Letnikov fractional derivative relates to the differences of a
fractional order $\Delta_h^{\mu}f(x)$ of a function $f(x)$ defined on the $x$ axis: $x\in R$. It can be presented as a finite difference with a step $h$ of order $\mu$ of the function $f(x)$
\begin{equation}\label{ch2_1-GL-1}
\Delta_h^{\mu}f(x)=\sum_{k=0}^{\infty}(-1)^k\binom{\mu}{k}f(x-kh)\, ,
\quad \mu>0\, .
\end{equation}
Then the function 
\begin{equation}\label{ch2_1-GL-2}
f^{\mu)}_{\pm}(x)=\lim_{h\to +0}\frac{(\Delta_{\pm h}f)(x)}{h^{\mu}}\, ,
\quad \mu>0
\end{equation}
is referred as the Gr\"unwald-Letnikov fractional derivative 
\cite{SaKiMa93}. It has been rigorously proven that
the expression \eqref{ch2_1-GL-2} coincides with the Marchaud derivative
with the same domain of definition \cite{SaKiMa93}(\S ~20.1), and it reads
\begin{equation}\label{ch2_1-GL-3}
f^{\mu)}_{\pm}(x) = D_{\pm}^{\mu}f(x)=
\lim_{\epsilon\to 0}\frac{\mu}{\Gamma(1-\mu)}\int_{\epsilon}^{\infty}
\frac{f(x)-f(x\mp y)}{y^{\mu+1}}dy \, , \quad 0<\mu<1\, ,
\end{equation}
and its symmetric form is
\begin{multline}\label{ch2_1-GL-4}
f^{\mu)}(x)=\frac{1}{2\cos\frac{\mu\pi}{2}}\lim_{h\to 0}
\frac{(\Delta_{+ h}f)(x)+ \Delta_{-h}f)(x)}{h^{\mu}} \\
=\frac{\mu}{2\Gamma(1-\mu)\cos\frac{\mu\pi}{2}}\int_{0+}^{\infty}
\frac{2f(x)-f(x+y)-f(x-y)}{y^{\mu+1}}dy\, ,
\end{multline}
which coincides with Eq. \eqref{ch2_1-RF-12}.
As it follows from Eqs. \eqref{ch2_1-RF-2} and \eqref{ch2_1-RF-3},
the fractional operator \eqref{ch2_1-GL-4} keeps $g(x)$ being periodic.

It is worth noting that the Marchaud derivatives $f^{\mu)}_{\pm}(x)$ in
Eq. \eqref{ch2_1-GL-3} are in the regularization form for $0<\mu<1$, while their symmetric combination in Eq. \eqref{ch2_1-GL-4} is regularized for $0<\mu<2$. Both the form do not destroy the periodicity of the exponential
$e^{i\bar{l}x}$.

\end{appendix}


\bibliography{data1}

\end{document}